\documentclass{article}
\usepackage{graphicx} 
\usepackage{xcolor}
\usepackage{spconf,amsmath,graphicx}
\usepackage{amsthm, amssymb}
\usepackage[hyperfootnotes=false]{hyperref}
\usepackage{marvosym}

\newcommand\blfootnote[1]{%
	\begingroup
	\renewcommand\thefootnote{}\footnote{#1}%
	\addtocounter{footnote}{-1}%
	\endgroup
}

\newtheorem{Rema}{Remark}[section]
\newtheorem{TheoPrinc}{Theorem}
\newtheorem{Conjecture}{Conjecture}

\newtheorem{prop}{Proposition}

\newcommand{\e}{\operatorname{e}}

\DeclareMathOperator{\v1}{V1}

\def\Z{{\mathbb Z}}    
\def\R{{\mathbb R}}

\title{REPRODUCING SENSORY INDUCED HALLUCINATIONS VIA NEURAL FIELDS} 
\name{Cyprien Tamekue\blfootnote{The first author was supported by a grant from the ``Fondation CFM pour la Recherche''.\\Email contact: firstname.lastname@l2s.centralesupelec.fr.}, Dario Prandi, and Yacine Chitour}
\address{\small{Université Paris-Saclay, CNRS, CentraleSupélec, Laboratoire des signaux et systèmes, 91190, Gif-sur-Yvette, France.}}
\begin{document}
%
\maketitle
\begin{abstract}
Understanding sensory induced cortical patterns in the primary visual cortex $\v1$ is an important challenge both for physiological motivations and for improving our understanding of human perception and visual organisation.
In this work we focus on pattern formation in the visual cortex when the cortical activity is driven by a geometric visual hallucination-like stimulus. In particular, we present a theoretical framework for sensory induced hallucinations which allows one to reproduce novel psychophysical results such as the MacKay effect (Nature, 1957) and the Billock and Tsou experiences (PNAS, 2007).\blfootnote{The first author was supported by a grant from the ``Fondation CFM pour la Recherche''.\\Email contact: firstname.lastname@l2s.centralesupelec.fr.}
\end{abstract}

\section{Introduction}
\label{sec:intro}

Spontaneous patterns forming in the primary visual cortex, which are the basis for visual hallucinations and illusions, inform us on the underlying mechanisms of human perception, allowing to refine its modelling and consequent implementation in various image processing or computer vision tasks.

In the pioneering works \cite{ermentrout1979mathematical, bressloff2001geometric}, the neural fields (NF) introduced in \cite{wilson1973mathematical} are used to theoretically describe spontaneous pattern formation in absence of external stimulus of geometric patterns in $\v1$ striate cortex ($\v1$ for short). Such patterns are the result of activity spreading over the field and correspond to paroxismic states of intrinsic cortical activity in $\v1$.
Taking into account the retinotipic correspondence between $\v1$ and the visual field \cite{tootell1982deoxyglucose,schwartz1977spatial}, these patterns are yields the geometric visual hallucinations (or form constants) classified by Klüver \cite{kluver1966mescal}, e.g., funnels, tunnels, spirals, checkerboards, cobwebs, etc...

However, these results lack to provide a description of more complex geometric visual hallucinations as those induced by flickering lights \cite{smythies1959stroboscopic}, or regular patterns with redundant informations \cite{mackay1958moving} and with visual noise \cite{mackay1961visual}. Such complex geometric visual hallucinations result from an external (visual) stimulus that disturbs the intrinsic activity in the field before the paroxismic state of cortical activity in $\v1$ occurs. In this work we investigate such hallucinations, by explicitly taking into account the external visual stimulus in the NF equations. 

The focus of this paper lies mainly in the MacKay effect, described in \cite{mackay1958moving,mackay1961visual} (see Figure~\ref{fig:mackay}).
The psycophysical experiences presented in these papers show that, e.g., presenting as a visual stimulus a funnel pattern (fan shape) with highly redundant information at the fovea, will induce the emergence in the visual field of a  ``complementary'' tunnel pattern (concentric circles) in the background, superimposed to the stimulus pattern. It is also reported that a tunnel pattern would produce a superimposed funnel pattern in the background as an after-image.

\begin{figure}
    \centering
    \includegraphics[width=.45\linewidth]{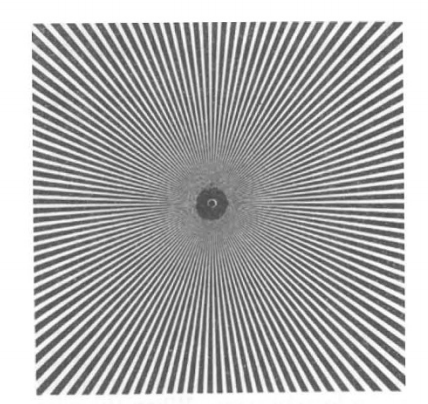}\hspace{1em}
    \includegraphics[width=.45\linewidth]{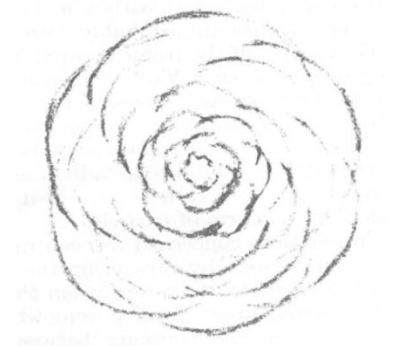}
    \caption{The MacKay effect: the presentation of the stimulus to the left induces the superimposed perception of the complementary image on the right. Reproduced from \cite{mackay1961visual}}
    \label{fig:mackay}
\end{figure}

Mathematically, we interpret Mackay effects as a controllability problem of NF equations describing cortical activity in $\v1$, where the control term is the external (visual) stimulus, which corresponds to the $\v1$ representation of the presented hallucination-like pattern. As a first step in the description of such effects, this work focuses on a simple model of $\v1\simeq \mathbb R^2$ based on a one-layer Amari-type equation \cite{amari1977dynamics} for the evolution of the cortical activity $a: \mathbb R^2\to \mathbb R$: 
\begin{equation}
\label{eq:nf}
{\partial_ta} = -\alpha a +\mu \, \omega\ast f(a)+\operatorname{I}.
\end{equation}
Here, $\alpha,\mu>0$ are parameters, $\ast$ denotes the spatial convolution operation, $\omega:\mathbb R^2 \to \mathbb R$ is an interaction (convolutional) kernel modelling cortical connections in $\v1$, $f$ is a sigmoid non-linearity, and $\operatorname{I}:\mathbb R^2 \to \mathbb R$ is the cortical representation of the presented visual stimulus. 

In our discussion, the choice of the parameters in \eqref{eq:nf} is such that, in absence of an external stimulus $\operatorname{I}$, the cortical state converges exponentially fast to the $0$ equilibrium. From a neurophysiological point of view, this means that we are considering an unaltered state, where no spontaneous geometric hallucination arises.  This prevents us from using bifurcation techniques as those employed in \cite{bressloff2001geometric}, and that allowed the authors of \cite{nicksUnderstanding2021} to obtain MacKay effects near the critical parameters for a variation of \eqref{eq:nf} incorporating a feedback effect in the external stimulus.

In this setting we first present a mathematically sound framework allowing us to define the input-output map $\Psi$, that associates to a time-invariant input stimulus $\operatorname{I}$ the stationary solution toward which \eqref{eq:nf} tends as $t\to +\infty$.
Then, we prove that highly redundant information in the centre of the funnel pattern is necessary for a MacKay effect to manifest. Namely, if the $\v1$ representation of the visual stimulus is a funnel pattern, then the output obtained via \eqref{eq:nf} will have the same shape. Following \cite{bressloff2001geometric} we represent such patterns as contrasting light and dark regions. 

Finally, we present numerical simulations obtained when the input patterns have highly redundant information in the sense of MacKay \cite{mackay1958moving}, showing that these reproduce the MacKay observations. We also reproduce visual hallucinations in the spirit of those recorded by Billock and Tsou \cite{billockNeural2007}. We stress that this is the first instance that these non-local phenomena are directly obtained via \eqref{eq:nf} (see \cite{nicksUnderstanding2021} for a different take on the question).

\section{Retino-cortical map}
One of the striking features of the functional architecture of $\v1$ striate cortex is its retinotopic organisation: (i) Neurons are organised in an orderly fashion called topographic or retinotopic mapping, in the sense that they form a 2D representation of the visual image formed on the retina in such a way that neighbouring regions of the image in the visual field are represented bijectively by neighbouring regions in $\v1$ area; (ii) Sufficiently close to the fovea (the centre of the visual field), the image has a much larger representation in V1 than in the visual field \cite{sereno1995borders} so that near the fovea, this map is an enlargement of the identity map; (iii) Far away from the fovea, this map is a log-polar transformation \cite{tootell1982deoxyglucose}.

This map is represented analytically as a complex logarithmic \cite{schwartz1977spatial}. After rescaling, it takes the form
\begin{equation}
   r e^{i\theta}  \mapsto (x_1,x_2):=\left( \log r, \theta \right).
\end{equation}
Here $(r,\theta)$ are retinal polar coordinates and $(x_1,x_2)$ are cortical cartesian cordinates.

A circle of radius $r$ and a ray of constant $\theta$ in the visual field become a vertical line and a horizontal line in $\v1$ respectively, see Figure~\ref{fig:retino}.

\begin{figure}
    \centering
    \includegraphics[width=.7\linewidth]{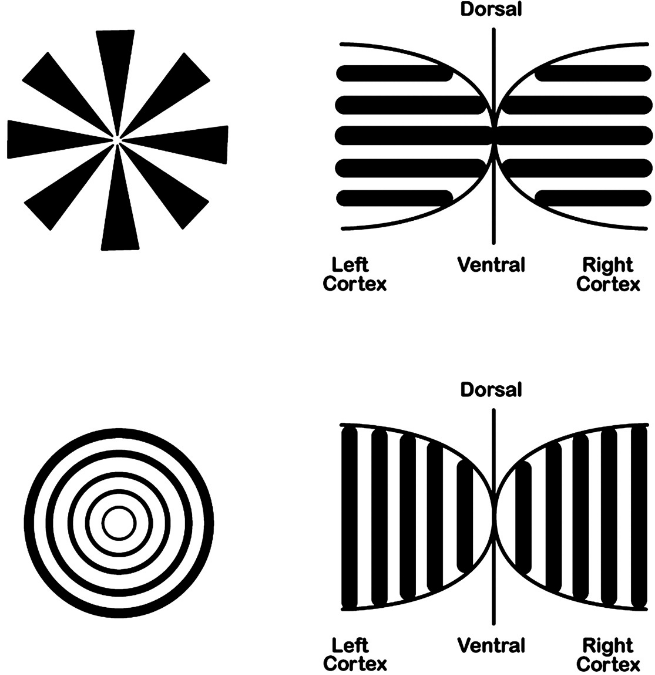}
    \caption{Visual illustration of the retino-cortical map. Reproduced from \cite{billockNeural2007}}
    \label{fig:retino}
\end{figure}

\section{Neural fields for cortical activity in V1}

We assume cortical activity in $\v1$ to be driven by equation \eqref{eq:nf}, where the non-linearity $f:\mathbb R\to \mathbb R$ is assumed to be a bounded $C^2$ odd function normalized such that $f'(0) = \max f'(t) = 1$ and $\|f\|_\infty = 1$. Following  \cite{ermentrout1979mathematical}, we also assume that the interaction kernel $\omega$ is radial (i.e., $\omega(x) = \omega(|x|)$), and that its Fourier transform $\hat \omega$ reaches its maximum at some non-zero critical wavenumber $q_c>0$. We call a vector $\xi_c\in\R^2$ satisfying $|\xi_c|=q_c$ a critical wavevector. This is the case, e.g., for the so-called Mexican hat distribution (or DoG, difference of Gaussians) but not for a simple Gaussian interaction kernel. 

It is shown in \cite{bressloff2001geometric} that under these assumptions there exists a critical value $\mu_c = \alpha / \hat\omega(q_c)$ such that equation \eqref{eq:nf} with $\operatorname{I}\equiv0$ is exponentially stable at $a\equiv 0$ if $\mu<\mu_c$. The \emph{spontaneous cortical patterns} are then obtained as the novel marginally stable equilibria appearing as bifurcations for $\mu=\mu_c$. These marginally stable equilibria are the elements in the kernel of the linear operator
\begin{equation}
    L_{\mu_c}a  = -\alpha a + \mu\,\omega\ast a.
\end{equation}
In particular, this kernel consists of linear combinations of functions of the form $\cos(2\pi\langle\xi^j,x\rangle)$ 
where, $\xi^j = q_c\e^{i\theta_j}$, $\theta_j\in\R$.

\section{Theoretical results}\label{s::theoretical}
Let $\mu_0:= \alpha / \|\omega\|_1 < \mu_c$.
As we shall see below, given $\operatorname{I}\in L^\infty(\R^2)$, $\mu_0$ is the (natural) largest value of $\mu$ up to which we can provide a stationary output to equation \eqref{eq:nf} in the space $L^\infty(\R^2)$.
From now on, we assume that the excitability rate $\mu$, of $\v1$ satisfies $\mu < \mu_0$. We first have the following.
\begin{prop}
Let $\mu<\mu_0$. Then, there exists a map $\Psi:L^\infty(\mathbb R^2)\to L^\infty(\mathbb R^2)$ such that for any initial datum $a_0\in L^\infty(\mathbb R^2)$ the solution to \eqref{eq:nf} with input $\operatorname{I}\in L^\infty(\mathbb R^2)$ converges exponentially to $\Psi(I)$. Moreover, $\Psi$ is bi-Lipschitz so that
\begin{equation}
    \Psi(\operatorname{I}) = 
    \operatorname{I}/\alpha+(\mu/\alpha)\omega\ast f(\Psi(\operatorname{I})).
\end{equation}
\end{prop}

\begin{proof}
Since the r.h.s.~of \eqref{eq:nf} is a Lipschitz map on $L^\infty_t(\mathbb R)\times L^\infty_x(\mathbb R^2)$, it is standard to obtain that any initial datum $a_0\in L^\infty(\mathbb R^2)$ admits a unique solution $a\in C([0,+\infty); L^\infty(\mathbb R^2))$.

The existence of a unique stationary solution (i.e., such that $\dot a = 0$) to \eqref{eq:nf} is obtained via a Banach fixed point argument. The map $\Psi$ is then defined as the map associating the input $I$ with this stationary solution. The exponential convergence to $\Psi(I)$ for any initial datum $a_0$ is obtained by the variation of constants formula and Gronwall's Lemma.
\end{proof}

\begin{Rema}\label{rmk::symmetry group}
Under assumptions on the kernel $\omega$, the map $v\mapsto\omega\ast f(v)$ commutes with  $\mathbf{E}(2)$, the group of rigid motion of the plane, see, \cite[Appendix A]{ermentrout1979mathematical}. It follows that a subgroup $\Gamma\subset\mathbf{E}(2)$ is a symmetry group of $\operatorname{I}\in L^\infty(\R^2)$ if and only if, it is a symmetry group of $\Psi(\operatorname{I})$.
\end{Rema}

We now focus on the \textit{funnel} and \textit{tunnel patterns} appearing in the MacKay effect. Taking into account the retino-cortical map, these patterns are generated by the cortical stimuli
\begin{equation}\label{eq::funnel and tunnel patterns}
	P_T(x) = \cos(\lambda x_1)\quad P_F(x) = \cos(\lambda x_2), \quad \lambda>0.
\end{equation}	
As a consequence of Remark~\ref{rmk::symmetry group}, one deduces that 
$\Psi(P_T)$ (resp. $\Psi(P_F)$) is a function of $x_1$ only (resp. $x_2$ only) even and $\pi/\lambda$ anti-periodic (one uses here the fact that $f$ is odd).

The following notation is needed: if $F$ is a real valued function, we use 
$F^{-1}(\{0\})$ to denote the set of zeroes of $F$.

\begin{TheoPrinc}\label{thm::main}
	Under the assumption $\mu<\mu_0/2$, we have $P_T^{-1}(\{0\}) = \Psi(P_T)^{-1}(\{0\})$ and $P_F^{-1}(\{0\}) = \Psi(P_F)^{-1}(\{0\})$.
\end{TheoPrinc}
\begin{proof}
We present an argument only for the first equality since the other one is obtained similarly by using the symmetries of Equation \eqref{eq:nf} and the fact that $\omega$ is radial.

Let us set $A_T = \Psi(P_T)$, for ease of notations. Then,
\begin{equation}
\label{eq:AT}
    A_T = P_T/\alpha + (\mu/\alpha) \omega\ast f(A_T).
\end{equation}
In particular
$A_T^{-1}(\{0\}) \supset \displaystyle\left\{\pm\pi/2\lambda 
+2k\pi \mid k\in\Z\right\}$.	

To show the converse inclusion, let $x_*:=(x_1^*,x_2^*)$ verifying $A_T(x_*) = 0$. From \eqref{eq:AT}, it follows 
	\begin{equation}\label{eq::useful}
	\cos(\lambda x_1^*) = -\mu\int_{\R^2}\omega(y)f(A_T(x_*-y))dy.
	\end{equation}
	On the one hand, by exploiting additions formulae for the cosine, one has for $\operatorname{a.e.}, y\in\R^2$,
	\begin{multline}\label{eq:x-y}
	A_T(x_*-y) =\sin(\lambda x_1^*)\sin(\lambda y_1)/\alpha\\
	+(\mu/\alpha)\int_{\R^2}k(y,z)f(A_T(x_*-z))dz,
	\end{multline}
	where $k(x,y) := \omega(x-y)-\cos(\lambda x_1)\omega(y),$ satisfies
\begin{equation}
K:=\sup\limits_{x\in\R^2}\int_{\R^2}|k(x,y)|dy= 2\|\omega\|_1.
\end{equation}
Since $\mu<\mu_0/2$, a Banach fixed point argument shows that for every $\operatorname{I}\in L^\infty(\R^2)$ there exists an unique solution $b\in L^\infty(\R^2)$ to 
\begin{equation}\label{eq::fixed point 1}
b(x) = \operatorname{I}(x)/\alpha+(\mu/\alpha)\int_{\R^2}k(x,y)f(b(y))dy.
\end{equation}
By \eqref{eq:x-y}, function $b(y):=A_T(x_*-y)$ is the unique solution  of the above equation associated to $I(y)=\sin(\lambda x_1^*)\sin(\lambda y_1)$. 

On the other hand, since $\omega$ is symmetric and the sigmoid $f$ is an odd function, we have also for $\operatorname{a.e.}, y\in\R^2$,
\begin{multline*}
	-A_T(x_*+y) =
	\sin(\lambda x_1^*)\sin(\lambda y_1)/\alpha\\
	+(\mu/\alpha)\int_{\R^2}k(y,z)f(-A_T(x_*+z))dz,
\end{multline*}
	so that, the function $\tilde b(y) = -b(-y)$ is also solution of equation \eqref{eq::fixed point 1} associated to $I(y)=\sin(\lambda x_1^*)\sin(\lambda y_1)$. By unicity of solution, one then has $b(-y) = -b(y)$ for $\operatorname{a.e.}, y\in\R^2$. This shows that $y\mapsto\omega(y)f(A_T(x_*-y))$ is an odd function on $\R^2$, since $\omega$ is symmetric and $f$ is an odd function, which implies that the r.h.s. of \eqref{eq::useful} is equal to $0$ and thus that $x^*\in P_T^{-1}(\{0\})$.
\end{proof}
The above theorem shows that, for our model of cortical activity in $\v1$, one cannot obtain a MacKay effect without perturbing the external input when chosen equal to $P_F$ or $P_T$. This does not contradicts the results of \cite{nicksUnderstanding2021}, where effects of MacKay type are derived directly with an input equal to $P_F$ or $P_T$ but for a different model of cortical activity in $\v1$. Indeed, the actual input injected in their model is not equal $P_F$ or $P_T$ but equal to $aP_F$ and $aP_T$, $a$ being the cortical activity .

We observe that $P_T^{-1}(\{0\})\subset \Psi(P_T)^{-1}(\{0\})$ is valid for $\mu<\mu_0$. 
Indeed, numerical simulations 
suggest that Theorem~\ref{thm::main} is valid even for $\mu_0/2\le\mu<\mu_0$.
The above remark motivates the following.
\begin{Conjecture}
	Let $P_T$ be the tunnel pattern defined in \eqref{eq::funnel and tunnel patterns}. Assume that the kernel $\omega$ is symmetric and the activation function $f$ is an odd function. Then, under Assumption $\mu<\mu_0$, $P_T$ and $\Psi(P_T)$ have the same set of zeroes.
\end{Conjecture}

\section{Numerical results}

According to the results of Section~\ref{s::theoretical}, for MacKay effects to occur it is necessary to break the $\mathbf{E}(2)$-symmetry of the stimulus pattern, say, $P_F$, via highly redundant information at the fovea. In our implementation we add to the stimulus pattern a localized perturbation $\varepsilon v$, 
where $\varepsilon>0$ and $v$ is a localized function, viz. $v(x) := \chi_{\Omega}(x)$, $\Omega\subset\R^2$.

The numerical implementation is performed in Julia \cite{bezanson2017julia}, and is available at \url{www.github.com/dprn/MacKay}. 
The operator $\Psi$ is numerically implemented via an iterative fixed-point method. The cortical data is defined on a square $[-L,L]^2$, with $L=10$. The interaction kernel is taken to be a DoG with variances $\sigma_1 = 0.1$ and $\sigma_2 = 0.5$, yielding $\mu_0 = 0.217$ and $\mu_c = 1.26$,  and the non-linearity is chosen as $f(t) = \tanh(t)$. The other parameters are chosen as $\alpha = 1$, $\mu = 0.99\mu_0$ and $\varepsilon = 0.025$. We stress that $\mu$ is quite far from the bifurcation point $\mu_c$.
\begin{figure}
    \centering
    \includegraphics[width = .95\linewidth]{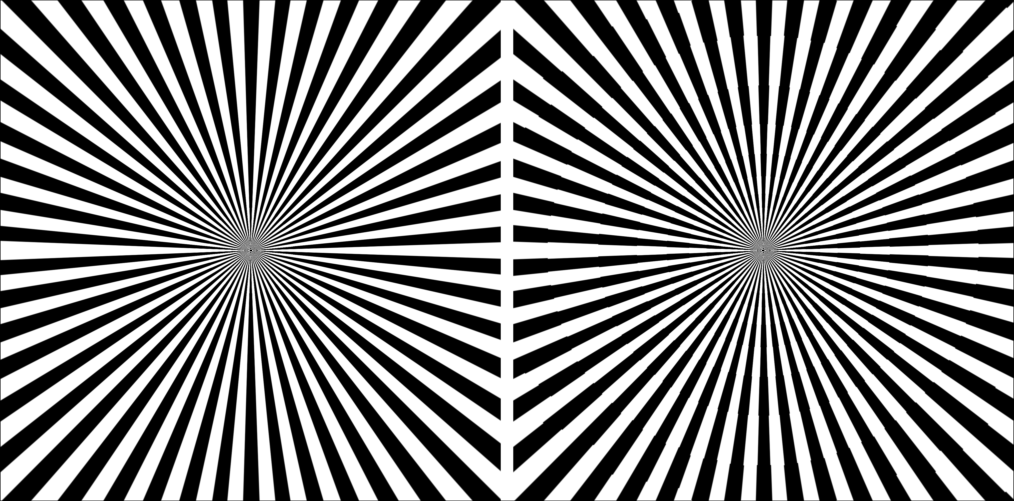}
    \caption{MacKay effect (\emph{right}) on a funnel-like pattern (\emph{left}). Initial input $I = P_F + \varepsilon\chi_\Omega$, $\Omega=[-L,2]\times[-L,L]$. Compare with Figure~\ref{fig:mackay}.}
    \label{fig:mackay-funnel}
\end{figure}

\begin{figure}
    \centering
    \includegraphics[width = .95\linewidth]{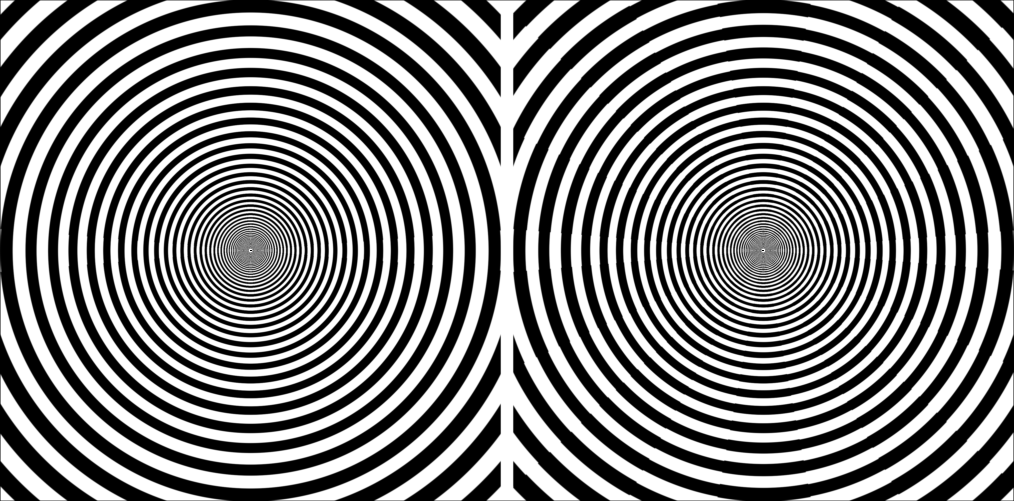}
    \caption{MacKay effect (\emph{right}) on a tunnel-like pattern (\emph{left}). Initial input $I = P_T + \chi_\Omega$, $\Omega=[-L,L]\times[-0.25,0.25]$.}
    \label{fig:mackay-tunnel}
\end{figure}

We collected some representative results in Figures~\ref{fig:mackay-funnel} and \ref{fig:mackay-tunnel}. Here, we visualize the input (left image) and the output (right image) in the retinal representation obtained from the cortical stimuli via the inverse retino-cortical map.
Also, following the convention adopted in \cite{ermentrout1979mathematical, bressloff2001geometric} for geometric visual hallucinations, we present binary versions of these images, where black corresponds to negative values and white to positive ones.

We also present Figure~\ref{fig:Bilock-funnel-fovea} and \ref{fig:Bilock-funnel-periphery}, where we show that our model can reproduce the psychophysiscal results of Billock and Tsou. We remark that these results are obtained with a different choice of nonlinearity $f(t) = (1+\exp(-t+0.25))^{-1}-(1+\exp(0.25))^{-1}$.

\begin{figure}[!ht]
    \centering
    \includegraphics[width = .95\linewidth]{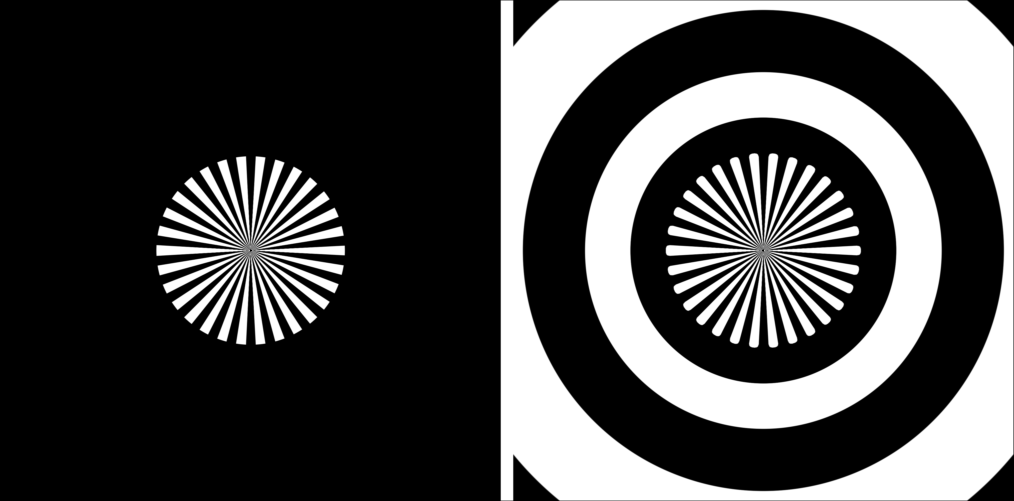}
    \caption{Billock and Tsou visual hallucination, with initial stimulus $I = P_F\chi_\Omega$, $\Omega=[-L,5]\times[-L,L]$. Compare with \cite[Fig 3., a]{billockNeural2007}.}
    \label{fig:Bilock-funnel-fovea}
\end{figure}
\vspace{-0.5cm}
\begin{figure}[!ht]
    \centering
    \includegraphics[width = .95\linewidth]{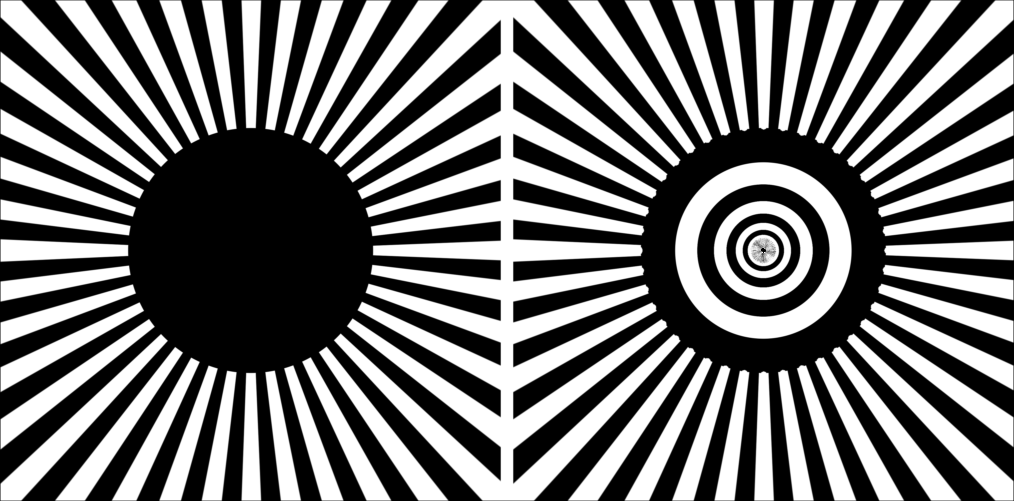}
    \caption{Billock and Tsou visual hallucination, with initial stimulus $I = P_F\chi_\Omega$, $\Omega=[5,L]\times[-L,L]$. Compare with \cite[Fig 3., c]{billockNeural2007}.}
    \label{fig:Bilock-funnel-periphery}
\end{figure}
\section{Conclusions}

In this paper we presented a model and theoretical framework for the description of sensory induced hallucinations. We also presented numerical experiences indicating the capability of this approach to reproduce the MacKay effects and psychophysical observations of \cite{billockNeural2007}. To our knowledge, this is the most parsimonious model to be able of such reproduction.

The numerical simulations provided in this work indicate also that the anisotropic nature of cortical connections in $\v1$ need not to be integrated in the model in order to reproduce the psychophysical results of \cite{billockNeural2007}. This is in contrast with the conjecture advanced by the authors of \cite{billockNeural2007}.

We highlight that both MacKay effects and psychophysical observations of \cite{billockNeural2007} result from input patterns, with global $\mathbf{E}(2)$-symmetry broken. In the former case, the symmetry is broken by localised highly redundant information in the input pattern, whereas, in the latter, the global symmetry is broken by localising the input pattern either in the left or in the right area of the cortex.

\bibliographystyle{IEEEbib}
\bibliography{tamekue2022}

\end{document}